\documentclass[a4paper,UKenglish]{lipics-v2021}
\usepackage{mathtools}
\usepackage{tikz}
\usetikzlibrary{calc}
\usetikzlibrary{positioning, shapes.geometric, fit}

\DeclarePairedDelimiter\abs{\lvert}{\rvert}
\usetikzlibrary{patterns}
\usetikzlibrary{arrows}

\DeclareMathOperator{\CFI}{CFI}
\DeclareMathOperator{\cl}{cl}

\newcommand{\If}{\hspace{6pt}\text{if}\hspace{6pt}}

\newcommand{\C}{\mathcal{C}}
\newcommand{\F}{\mathcal{F}}
\newcommand{\M}{\mathcal{M}}
\newcommand{\imm}{\preccurlyeq_{\text{imm}}}
\newcommand{\immsimp}{\preccurlyeq_{\text{immsimp}}}
\newcommand{\N}{\mathbb{N}}

\DeclareMathOperator{\id}{id}

\newcommand\restr[2]{{%
  \left.\kern-\nulldelimiterspace %
  #1 %
  \vphantom{\big|} %
  \right|_{#2} %
  }}

\newtheorem{question}[theorem]{Question}

\newenvironment{proofofclaim}[1][Proof of claim]{\begin{claimproof}[#1]}{\end{claimproof}}
\nolinenumbers
\hideLIPIcs
\makeatletter
\let\@oddfoot\@empty
\makeatother
\AtBeginDocument{\hypersetup{pdfsubject={Oddomorphisms, Split-Off Minors, and the Strong Roberson Conjecture}}}

\title{Oddomorphisms, Split-Off Minors, and the Strong Roberson Conjecture}
\titlerunning{Oddomorphisms, Split-Off Minors, and Roberson's Conjecture}
\author{Arnar Á. Kristjánsson}{University of Cambridge, United Kingdom}{}{https://orcid.org/0009-0002-4483-4667}{}
\authorrunning{A. Á. Kristjánsson}
\Copyright{Arnar Á. Kristjánsson}

\category{}
\supplement{}
\funding{Research funded in part by UK Research and Innovation (UKRI) under the UK government’s Horizon Europe funding guarantee: grant number EP/X028259/1.}

\acknowledgements{I am grateful to my supervisor, Anuj Dawar, for his guidance throughout this research, for frequent insightful conversations on the topic of the paper, and for his valuable feedback on a draft of the paper.
I am also grateful to Tim Seppelt and David Roberson for helpful discussions on the topic, and for reading a preliminary version of Section~\ref{sec:oddo-doesnt-mean-minor}.
Finally, I want to thank the anonymous reviewers for their careful reading of the manuscript and their helpful suggestions.
}

\begin{document}

\maketitle

\begin{abstract}
We show that the existence of an oddomorphism from a graph $F$ to a graph $G$ does not imply that $G$ is a minor of $F$. This answers a question posed by Roberson~\cite{roberson2022oddomorphisms} and shows that the CFI graphs cannot be used to prove the Strong Roberson Conjecture.
Additionally, we introduce the concept of a \emph{split-off minor} and show that the existence of an oddomorphism from $F$ to $G$ implies that $G$ is a split-off minor of $F$. Consequently, every class that is closed under taking split-off minors and disjoint unions is homomorphism distinguishing closed. The split-off minor relation is the first minor-like structural relation shown to have this property, marking a meaningful advancement in our understanding of the interaction between structural graph containment and homomorphism indistinguishability relations.
\end{abstract}

\section{Introduction}
A central technique in the study of limitative results and lower bounds in theoretical computer science is the construction of a pair of distinct structures that appear nearly identical.
To this end, the Cai–Fürer–Immerman (CFI) construction has established itself as an exceptionally versatile tool.
It produces, from a given connected base graph $G$, a pair of graphs $\CFI_0(G)$, $\CFI_1(G)$, which are locally similar but non-isomorphic.
The construction was introduced by Cai, Fürer, and Immerman \cite{cai1992optimal} to show that, for any $k$, $k$-variable counting logic cannot distinguish all non-isomorphic graphs.
Equivalently, this demonstrates that the $k-1$-dimensional Weisfeiler-Leman algorithm fails to solve the graph isomorphism problem.
Subsequently, the construction has been adapted in numerous forms across a wide range of problems, becoming ubiquitous in finite model theory.
For example, variations of it have repeatedly been used to separate logics from \textsf{PTIME} \cite{Dawar2007thepowerofcounting,gradel2019rank,lichter2023separating} and in the analysis of the Weisfeiler-Leman algorithm \cite{Furer2001weisfeiler, FUHLBRUCK202196, grohe2025compressing}. Recently, it has also been used to establish a \textsf{\#P}-hardness result \cite{DBLP:conf/soda/Curticapean24} and to prove the equirank homomorphism preservation theorem \cite{rossman:LIPIcs.CSL.2025.6}.
\par 
Another line of research that has been gaining traction in recent years is the use of homomorphism counts for characterizing relations between structures. The foundational result by Lovász \cite{lovasz1967operations} states that two finite graphs are isomorphic if and only if they have the same number of homomorphisms from all finite graphs. Dvořák~\cite{dvovrak2010recognizing} later proved that a pair of graphs is indistinguishable in $k$-variable counting logic if and only if they have the same number of homomorphisms from every graph of treewidth at most $k-1$, and Grohe \cite{grohe2020counting} proved an analogous result relating counting logic with quantifier rank bounded by $k$ to homomorphism counts from graphs of treedepth at most $k$. These results were generalized categorically by Dawar, Jakl, and Reggio \cite{Dawar2021Lovasztype} and are now often referred to as \emph{Lovász-type theorems}.
Another important Lovász-type theorem that does not fit into the categorical generalisation is the result by Mančinska and Roberson \cite{Mancinska2020:quantum} that a pair of graphs is quantum isomorphic if and only if they have the same number of homomorphisms from all planar graphs. 
\par 
These results led to the study of \emph{homomorphism indistinguishability relations} more generally. For a class $\F$ of graphs, its homomorphism indistinguishability relation $\equiv_\F$ consists of the pairs of graphs that have the same homomorphism count from every element of $\F$. Roberson \cite{roberson2022oddomorphisms} posed the question of when two such relations $\equiv_{\F_1}$, $\equiv_{\F_2}$ are distinct. 
For some graphs $F$, the homomorphism counts from $F$ to all graphs $G$ are already determined by the homomorphism counts from graphs in $\F$. Equivalently, $\equiv_{\F \cup\{F\}} = \equiv_\F$. By adding all such graphs $F$ to the class $\F$, we obtain the \emph{homomorphism distinguishing closure} $\cl(\F)$ of $\F$, defined by Roberson. A class $\F$ with $\cl(\F) = \F$ is said to be \emph{homomorphism distinguishing closed}, abbreviated \emph{h.d.\ closed}.
Roberson observed that, if $\F_1$ and $\F_2$ are h.d.\ closed, inclusion between their indistinguishability relations corresponds to reverse inclusion between the classes.
This motivated the study of which classes are h.d.\ closed and the following conjecture.
\begin{conjecture}[\emph{Strong Roberson Conjecture}, Conjecture 4 in \cite{roberson2022oddomorphisms}]
    Every class of graphs that is minor closed and closed under disjoint unions is homomorphism distinguishing closed.
\end{conjecture}
If true, the conjecture would let us compare the homomorphism indistinguishability relations of classes closed under minors and disjoint unions simply by checking inclusion between the underlying classes.\par 
In his work and in subsequent research, numerous classes have been shown to be h.d.\ closed, including many of the canonical minor closed classes. This includes the class of planar graphs \cite{roberson2022oddomorphisms}, the class of graphs of treewidth at most $k$ \cite{neuen2023homomorphism}, and the class of graphs of treedepth at most $k$ \cite{fluck2023goingdeepgoingwide}.
Roberson also posed a weaker conjecture that would still be a powerful result.
\begin{conjecture}[\emph{Weak Roberson Conjecture}, Conjecture 5 in \cite{roberson2022oddomorphisms}]
    Every class of graphs that is minor closed, closed under disjoint unions, and is not the class of all graphs has a homomorphism indistinguishability relation that is not equal to isomorphism.
\end{conjecture}
\par 
The main contribution of Roberson's work \cite{roberson2022oddomorphisms} was the introduction and analysis of the concept of an \emph{oddomorphism}. These are graph homomorphisms that satisfy certain parity constraints on their fibres (see Definition \ref{def:oddomorphisms}). The reason for our interest in oddomorphisms is that they delineate which graphs can distinguish a pair of CFI graphs with homomorphism counts.
\begin{theorem}[Theorem 3.13 in \cite{roberson2022oddomorphisms}]
\label{thm:fundamental-oddo}
    For a connected graph $G$, the numbers of homomorphisms from a graph $F$ to $\CFI_0(G)$ and $\CFI_1(G)$ are different if and only if $F$ has a weak oddomorphism to $G$.
\end{theorem}

The \emph{closure} of a graph class $\F$ under weak oddomorphism is the class of all graphs $G$ such that there exists a graph $F \in \F$ and a weak oddomorphism $F \to G$. We say that $\F$ is \emph{closed under weak oddomorphisms} if its closure under weak oddomorphisms is equal to $\F$.

Theorem~\ref{thm:fundamental-oddo} allows for a general study of $\CFI$ graphs through the framework of homomorphism indistinguishability. By only observing the weak oddomorphisms, we can determine which $\CFI$ graph pairs are homomorphism indistinguishable with respect to which classes of graphs. 

As an example, using Theorem~\ref{thm:fundamental-oddo} and the previously mentioned results by Dvořák~\cite{dvovrak2010recognizing}, we can cast the original result of Cai, Fürer, and Immerman~\cite{cai1992optimal} into this framework as the following statement: the closure of the class of graphs of treewidth at most $k$ under weak oddomorphism is not equal to the class of all graphs. 
Namely, homomorphism indistinguishability over graphs of treewidth at most $k$ coincides with equivalence in $(k+1)$-variable counting logic.
By Theorem~\ref{thm:fundamental-oddo}, a connected graph $H$ lies outside this closure if and only if $\CFI_0(H)$ and $\CFI_1(H)$ are indistinguishable in this logic.

This result also gives a neat method for proving that classes are h.d.\ closed. Namely, the classes that are closed under weak oddomorphisms are precisely the classes whose closure under disjoint unions and restrictions to connected components can be proven to be h.d.\ closed using $\CFI$ graphs. 
\begin{theorem}[Theorem 6.2 in \cite{roberson2022oddomorphisms}]
\label{thm:oddo-implies-hdclosed}
    A class of graphs that is closed under disjoint unions, restrictions to connected components, and weak oddomorphisms is h.d.\ closed.
\end{theorem}
To prove the Strong Roberson Conjecture, it therefore suffices to answer the following question in the positive.
\begin{question}[Question 5 in \cite{roberson2022oddomorphisms}]
\label{q}
    Does the existence of a weak oddomorphism from $F$ to $G$ imply that $F$ contains $G$ as a minor for any connected graph $G$?
\end{question}
To attest to the power of CFI graphs, almost all classes known to be h.d.\ closed are closed under weak oddomorphisms, with the notable exception of essentially profinite classes \cite{Seppelt2024:logical}.
For this reason, answering Question \ref{q} in the positive has been considered to be the only viable path to proving the Strong Roberson Conjecture. 
\par
With the exception of Roberson's original work \cite{roberson2022oddomorphisms}, many of the results proving that classes are closed under weak oddomorphisms \cite{neuen2023homomorphism, fluck2023goingdeepgoingwide, schindling:LIPIcs.MFCS.2025.89} were not proved using combinatorial or structural properties of oddomorphisms, but rather adaptations of the game arguments from \cite{Dawar2007thepowerofcounting}.
However, two recent pieces of work have made significant improvements to our structural understanding of oddomorphisms and homomorphism indistinguishability relations. 
First, Neuen and Seppelt \cite{neuen2026distinguishinggraphscountinghomomorphisms} used oddomorphisms to prove that the homomorphism indistinguishability relation of every class of bounded \emph{vortex-free Hadwiger number} is not equal to isomorphism. 
Additionally, they proved that there exists a class of graphs that excludes a fixed graph as a topological minor and has a homomorphism indistinguishability relation that is equal to isomorphism. 
Since a graph class with a bounded vortex-free Hadwiger number excludes some graph as a minor, and minor-closed classes in turn exclude topological minors, these results delineate the boundaries of the Weak Roberson Conjecture. Specifically, they establish that a restricted version of the conjecture holds, while a broader generalization fails.
Second, Jiménez, Moore, Quiroz, and Yoo~\cite{jiménez2026homomorphismcountingimmersionclosedclasses}
introduced the idea that perhaps \emph{graph immersions} have a more suitable connection to oddomorphisms and h.d.\ closedness than minors do. A graph $F$ contains a graph $G$ as an immersion if $G$ can be obtained from $F$ using the operations of taking a subgraph and \emph{splitting off} a path of length 2, i.e. deleting it and adding an edge between its endpoints. This gives a relation that is strictly weaker than the topological minor but incomparable to the minor.
They used oddomorphisms to prove the analogue of the Weak Roberson Conjecture with minors replaced with immersions. That is, they proved that every class that excludes a fixed graph as an immersion has a homomorphism indistinguishability relation that is not equal to isomorphism. 
This is neither a weaker nor a stronger statement than the original Weak Roberson Conjecture since the graph immersion relation is incomparable to the graph minor relation.
\par 

The main contributions of this paper are two results about oddomorphisms. 
First, we resolve Question \ref{q} in the negative. 
That is, we construct a pair of graphs $F$, $G$ such that $F$ has an oddomorphism to $G$ but $G$ is not a minor of $F$.
We give two different methods for obtaining such a pair, each with its own appeal. One is a deterministic construction of a specific pair of rather small graphs $F$ and $G$ satisfying the desired properties. The other leverages randomness to non-constructively find infinitely many such pairs.
This result implies that oddomorphisms and CFI graphs cannot be used to prove the Strong Roberson Conjecture. As previously mentioned, this was considered the only known feasible path to proving it.
The result therefore shifts the focus of research aiming to resolve the conjecture.
\par 
Second, we introduce the concept of a \emph{split-off minor} and prove a tight connection between oddomorphisms and split-off minors.
\begin{theorem}
\label{thm:oddo-implies-lm}
    If a graph $F$ has an oddomorphism to a graph $G$, then $G$ is a split-off minor of $F$.
\end{theorem}
The split-off minor relation is in some sense the smallest natural relation that extends both the graph minor and immersion relations. 
A graph $G$ is a split-off minor of $F$ if $G$ can be obtained from $F$ by applying the operations of taking subgraphs, contracting edges, and splitting off paths of length 2.
Theorem \ref{thm:oddo-implies-lm} implies that every class closed under split-off minors is closed under weak oddomorphisms, and therefore, if such a class is additionally closed under taking disjoint unions it is h.d.\ closed.
This proves the Strong Roberson Conjecture with minors replaced with split-off minors.
This result aligns well with the flexibility of the framework introduced by Roberson when formulating the conjecture. Indeed, he explicitly noted that the graph minor relation might ultimately need to be replaced by another property for the conjecture to hold.
This is the first result that proves such a modified Strong Roberson Conjecture with the minor condition replaced with another similar structural relation. It is therefore a significant step towards a better understanding of how the homomorphism indistinguishability relation of a class of graphs is affected by its structural properties.
\par 
The rest of the paper is dedicated to exploring properties of the split-off minor relation. We show that it is distinct from the closely related \emph{lift-minor} relation studied in~\cite{GOLOVACH2014286}. We then give a few examples of classes that are closed under taking split-off minors, giving applications of Theorem \ref{thm:oddo-implies-lm} to finding h.d.\ closed classes.
Some of these classes had already been proven to be h.d.\ closed by Neuen and Seppelt~\cite{neuen2026distinguishinggraphscountinghomomorphisms}, such as the class of \emph{cactus graphs} and every class of forests closed under taking minors and disjoint unions.
However, we also obtain a new family of examples by showing that for every natural number $n$, the class of graphs whose connected components have \emph{cyclomatic number} at most $n$ is h.d.\ closed.

\section{Preliminaries}
\subsection{Operations on Graphs and Multigraphs}
Unless otherwise stated, a \emph{graph} is finite, simple, undirected, and loopless. We denote the set of vertices and edges of a graph $G$ with $V(G)$ and $E(G)$, respectively.
Multigraphs are defined analogously, except that $E(G)$ is a multiset and may contain loops, which we can represent as a function $\chi_{E(G)}: \{ \{u,v\} \mid u,v \in V(G) \} \to \N_0$. We view standard graphs as a subclass of multigraphs. For definitions of basic graph theoretic notions not introduced here, we refer the reader to \cite{diestel2025graph}.

\par 
When an edge $e$ with distinct endpoints $v$ and $w$ is contracted in a multigraph $G$, the graph $G'$ is formed.
Its set of vertices is defined by 
\[
V(G') \coloneqq (V(G)\setminus \{v,w\}) \sqcup \{e\}
\]
where $\sqcup$ denotes the disjoint union. The edge set is defined by \begin{align*}
    \chi_{E(G')}(\{u,u'\}) \coloneqq \begin{cases}
        \chi_{E(G)}(\{u, u'\}) &\If u,u' \neq e \\
        \chi_{E(G)}(\{v, u'\}) + \chi_{E(G)}(\{w, u'\}) &\If u = e \neq u'\\
        \chi_{E(G)}(\{v,w\}) -1 + \chi_{E(G)}(\{v\}) + \chi_{E(G)}(\{w\}) &\If u = e = u'
    \end{cases}.
\end{align*}
If $G$ is simple, the contraction operation can also be considered as a simple-graph operation. In this case, multi-edges are merged into a single edge and self-loops are removed after the contraction.
When considered as a multigraph operation, contracting an edge might turn a simple graph into a non-simple graph.
\par 
The operation of \emph{splitting off} (also called lifting) can be applied to a graph containing a pair of edges $e,e'$ with endpoints $u,v$ and $v,w$, respectively. The resulting graph is obtained by deleting the edges $e$ and $e'$ and adding an edge between $u$ and $w$. The vertices $u$ and $w$ may coincide, in which case the added edge is a loop.
Again, this can be considered as both a simple-graph operation and a multigraph operation, where, in the former case, the graph is made simple after the application of the operation. Note that if there is a path $P$ going through the vertices $v_1,\ldots, v_n$, we can repeatedly split off the edges in the path until all of them have been deleted and an edge between $v_1$ and $v_{n}$ has been added. We call this \emph{splitting off the path $P$}.
\par
We can now define graph minors and immersions.
\begin{definition}
    A multigraph $H$ is a \emph{minor} of a multigraph $G$ if $H$ can be obtained from $G$ by applying the operations of contracting edges (considered as a multigraph operation) and taking subgraphs.
\end{definition}

\begin{definition}
\label{def:immersions}
    A multigraph $H$ is said to be \emph{immersed} in a multigraph $G$ if $H$ can be obtained from $G$ by applying the operations of splitting off edges (considered as a multigraph operation) and taking subgraphs. In this case we also say that $G$ \emph{contains $H$ as an immersion} or that \emph{$H$ is an immersion of $G$}.
\end{definition}

Both minors and immersions also have characterisations that avoid graph operations.
\begin{proposition}
    For a graph $H$ and a multigraph $G$, the following are equivalent:
    \begin{enumerate}[(i)]
        \item $H$ is a minor of $G$.
        \item There exists a surjective partial function $\pi:V(G) \rightharpoonup V(H)$ such that each fibre $\pi^{-1}(v)$ for $v\in V(H)$ is connected in $G$ and for each edge $\{v,w\}$ in $H$ there exists an edge in $G$ that maps to $\{v,w\}$ under $\pi$ (when extended to a function on the edges).
    \end{enumerate}
\end{proposition}

The characterisation can be adapted to multigraphs $H$, but this version is more convenient for our purposes in Section~\ref{sec:first-counterexample}. 

\begin{proposition}
    For a loopless multigraph $H$ and a multigraph $G$, the following are equivalent:
    \begin{enumerate}[(i)]
        \item $H$ is immersed in $G$.
        \item There exists an injective function $\pi:V(H) \to V(G)$ and a set of edge-disjoint paths $\{P_e \mid e \in E(H)\}$ in $G$ such that $P_e$ has endpoints $\pi(v)$ and $\pi(w)$, where $v$ and $w$ are the endpoints of $e$.
    \end{enumerate}
\end{proposition}
If $\pi$ is a function witnessing the immersion of $H$ in $G$, then the vertices in the image of $\pi$ are said to be the \emph{branch vertices} of the immersion.

\subsection{Oddomorphisms}
Oddomorphisms were introduced by Roberson \cite{roberson2022oddomorphisms}. We give a slightly different but equivalent definition here.
\begin{definition}
\label{def:oddomorphisms}
    A homomorphism $\varphi$ from a multigraph $F$ to a graph $G$ is an \emph{oddomorphism} if every vertex of $F$ can be labelled either odd or even such that 
    \begin{enumerate}[(i)]
        \item There is an odd number of odd vertices in each fibre of $\varphi$.
        \item For each $e \in E(G)$, each odd vertex has odd degree and each even vertex has even degree in the induced subgraph of $F$ on the vertices in $\varphi^{-1}(e)$.
    \end{enumerate}
    If a vertex is odd/even under this labelling, we say it is \emph{$\varphi$-odd/even} or \emph{odd/even with respect to $\varphi$}.\par 
    A homomorphism $\varphi: F \to G$ is a \emph{weak oddomorphism} if there is a subgraph $F'$ of $F$ such that the restriction $\restr{\varphi}{F'}: F' \to G$ is an oddomorphism.
\end{definition}

\begin{figure}[htbp]
    \centering
    
\begin{tikzpicture}[
    vertex/.style={circle, draw=black, thick, minimum size=14pt, inner sep=0pt, font=\small},
    odd/.style={vertex, fill=blue!10},
    even/.style={vertex, fill=white},
    edge/.style={thick, black},
    bag/.style={ellipse, draw=gray, dashed, inner sep=8pt}
]

    \node[odd]  (Ao)  at (0.6, 5.2)   {o};
    \node[even] (Ae1) at (0, 4) {e};
    \node[even] (Ae2) at (0.8, 4)  {e};
    \node[bag, fit=(Ao) (Ae1) (Ae2)] (BagA) {};

    \node[odd]  (Bo)  at (-3, 0.5)   {o};
    \node[even] (Be1) at (-2.5, 2) {e};
    \node[even] (Be2) at (-1.5, 1) {e};
    \node[bag, fit=(Bo) (Be1) (Be2)] (BagB) {};

    \node[odd]  (Co)  at (3, 0.5)   {o};
    \node[even] (Ce1) at (2, 0) {e};
    \node[bag, fit=(Co) (Ce1)] (BagC) {};

    \draw[edge] (Ao)  -- (Be1);
    \draw[edge] (Be1) -- (Ae1);
    \draw[edge] (Ae1) -- (Be2);
    \draw[edge] (Be2) -- (Ae2);
    \draw[edge] (Ae2) -- (Bo);

    \draw[edge] (Bo)  -- (Ce1);
    \draw[edge] (Ce1) -- (Be2);
    \draw[edge] (Be2) -- (Co);

    \draw[edge] (Co) -- (Ao);

\end{tikzpicture}
    \caption{An example of a graph with an oddomorphism to $K_3$. The fibres of the oddomorphism are indicated by the dashed lines.}
    \label{fig:oddo_example}
\end{figure}

In condition (ii), the induced subgraph on the vertices in $\varphi^{-1}(e)$ is equal to the multigraph with vertex set $\varphi^{-1}(e)$ and edge multiset $\{ f \mid \varphi(f) = e \}$.

\section{The Strong Roberson Conjecture Cannot Be Proved Using Oddomorphisms}
\label{sec:oddo-doesnt-mean-minor}
In this section we resolve Question \ref{q} in the negative.
In fact, we give two different constructions showing that the existence of an oddomorphism (and thus a weak oddomorphism) from $F$ to $G$ need not imply that $F$ contains $G$ as a minor. The first is a new construction while the second is drawn from the literature. Identifying it as a counterexample requires only a minor observation.

\subsection{First Example: Gluing of 5-Cycles}
\label{sec:first-counterexample}
The construction of this example requires some preparation before it can be presented.
A key step is to transfer the problem to the more restrictive setting of rooted minors (where the minor map must preserve specified vertex labels).
Then it suffices to find examples of graphs $G$ and $F$ such that $G$ is not a rooted minor of $F$ with some prescribed labels.
\par 
For graphs $G$ and $H$ with $v$ a vertex in $G$ and $h$ a vertex in $H$ we let $G_v \oplus H_h$ denote the gluing of $G$ and $H$ on the vertices $v$ and $h$. More specifically, $G_v \oplus H_h$ can be constructed by taking the disjoint union of $G$ and $H$, adding the edge $\{v,h\}$ and then contracting it.
\par
The following lemma is an application of a fairly standard technique in structural graph theory.
\begin{lemma}
\label{lemma:rooted-minors}
    Let $F$ and $G$ be connected graphs and let $v_1,\ldots, v_n$ be distinct vertices in $G$ and $w_1,\ldots, w_n$ be distinct vertices in $F$.
    The graphs $F$ and $G$ can be extended to graphs $F_+$ and $G_+$ which satisfy the following: $G_+$ is a minor of $F_+$ if and only if there is a partial function $\pi: V(F) \rightharpoonup V(G)$ witnessing the fact that $G$ is a minor of $F$ and additionally satisfying $\pi(w_i) = v_i$ for each $i \in \{1,\ldots, n\}$. Moreover, $F_+$ and $G_+$ are obtained by gluing, for each $i$, the centres of two copies of a star to $w_i$ in $F$ and $v_i$ in $G$, respectively.
\end{lemma}
\begin{proof}
    Let $D$ be the sum of the degrees of the vertices in $V(F)$. We construct $F_+$ from $F$ by adding $D\cdot 2^i$ degree 1 vertices adjacent to the vertex $w_i$ for each $i$. Similarly we construct $G_+$ from $G$ by adding $D\cdot 2^i$ degree 1 vertices adjacent to $v_i$ for each $i$.
    If $\pi$ witnesses the fact that $G$ is a minor of $F$ and $\pi(w_i) = v_i$ for each $i$, we can extend $\pi$ to a witness of the fact that $G_+$ is a minor of $F_+$ by bijectively mapping the new degree 1 vertices.
    \par
    For the other direction, let $\pi_+: V(F_+) \rightharpoonup V(G_+)$ witness the fact that $G_+$ is a minor of $F_+$.
    For each $i$, the sum of the degrees of the vertices in $\pi_+^{-1}(v_i)$ must be at least $D\cdot 2^i$.
    We take as an induction hypothesis that $w_j \in \pi_+^{-1}(v_j)$ for each $j>i$. The sum of the degrees of the remaining vertices excluding $w_i$ is $\sum_{0\leq j<i} D \cdot 2^j < D\cdot 2^i$, so it follows that $w_i$ must belong to the fibre $\pi_+^{-1}(v_i)$. Using induction we see that this holds for all $i\in \{1,\ldots, n\}$.
    Additionally, if $\pi_+$ maps a degree 1 vertex that is adjacent to $w_i$ and is outside of $F$ to a vertex inside $V(G)\setminus \{v_i\}$, that vertex must be a degree 1 vertex in $G_+$ that is adjacent to $v_i$. It is easy to see that such vertices can be permuted to define a partial function $\pi_+'$ that is equal to $\pi_+$ everywhere except at those degree 1 vertices and whose restriction to $V(F)$ is onto $V(G)$. This restriction is then the witness $\pi$ of the fact that $G$ is a minor of $F$ which additionally satisfies $\pi(w_i) = v_i$ for each $i$.
\end{proof}

To take advantage of Lemma \ref{lemma:rooted-minors}, we must show that oddomorphisms $F \to G$ can be extended to the extensions $F_+$ and $G_+$.
We show the slightly more general fact that they can indeed be extended by gluing a graph to both the domain and codomain. 

\begin{lemma}
\label{lemma:oddo_extension}
    Let $\psi: F \to G$ be an oddomorphism, $v$ be a vertex of $G$, and let $H$ be a graph with a vertex $h$. If there is exactly one $\psi$-odd vertex $u$ in $\psi^{-1}(v)$, then $\psi$ extends to an oddomorphism $\psi_+:F_u \oplus H_h \to G_v \oplus H_h$ with $\psi_+(x)=\psi(x)$ for $x\in V(F)$ and $\psi_+(x)=x$ for $x\in V(H)\setminus\{h\}$.
\end{lemma}
\begin{proof}
    In the new oddomorphism, all of the new elements added from $H$ in $F_u \oplus H_h$ are labelled odd, and the others keep their labels.
    The number of odd vertices in a fibre $\psi_+^{-1}(a)$ for $a \in V(G)$ is unaffected by the gluing. Likewise, for $e \in E(G)$, the induced subgraph of $F_u \oplus H_h$ on $\psi_+^{-1}(e)$ is the same graph as the induced subgraph of $F$ on $\psi^{-1}(e)$. The degree condition therefore also holds for these edges.
\par
    Each fibre $\psi_+^{-1}(a)$ for $a \in V(H) \setminus \{h\}$ contains precisely one element, so the condition that each fibre contains an odd number of odd elements is also satisfied for these vertices. To see that the degree condition is satisfied on the $\psi_+$-preimages of edges $e$ in $H$ we note that the induced subgraph of $F_u \oplus H_h$ on $\psi_+^{-1}(e)$ is just a single edge between the corresponding odd vertices. So it is clear that each odd vertex has degree 1 in this graph while an even vertex has degree 0.
\end{proof}

\begin{figure}[ht]
    \centering
    \begin{tikzpicture}[
    vertex/.style={circle, draw=black, thick, minimum size=14pt, inner sep=0pt, font=\small},
    odd/.style={vertex, fill=blue!10},
    even/.style={vertex, fill=white},
    edge/.style={thick, black},
    bag/.style={ellipse, draw=gray, dashed, inner sep=8pt},
]

    \node[odd]  (Ao)  at (0.6, 5.2)  {o};
    \node[even] (Ae1) at (0, 4)     {e};
    \node[even] (Ae2) at (0.8, 4)   {e};
    \node[bag, fit=(Ao) (Ae1) (Ae2)] (BagA) {};

    \node[odd]  (Bo)  at (-3, 0.5)   {o};
    \node[even] (Be1) at (-2.5, 2)  {e};
    \node[even] (Be2) at (-1.5, 1)  {e};
    \node[bag, fit=(Bo) (Be1) (Be2)] (BagB) {};

    \node[odd]  (Co)  at (3, 0.5)    {o};
    \node[even] (Ce1) at (2, 0)      {e};
    \node[bag, fit=(Co) (Ce1)] (BagC) {};

    \node[odd] (C4a) at (5.2, 0.5)  {o};
    \node[odd] (C4b) at (5.2, -1.7)   {o};
    \node[odd] (C4c) at (3, -1.7)     {o};
    
    \node[bag, fit=(C4a)] {};
    \node[bag, fit=(C4b)] {};
    \node[bag, fit=(C4c)] {};

    \draw[edge] (Ao)  -- (Be1);
    \draw[edge] (Be1) -- (Ae1);
    \draw[edge] (Ae1) -- (Be2);
    \draw[edge] (Be2) -- (Ae2);
    \draw[edge] (Ae2) -- (Bo);

    \draw[edge] (Bo)  -- (Ce1);
    \draw[edge] (Ce1) -- (Be2);
    \draw[edge] (Be2) -- (Co);

    \draw[edge] (Co) -- (Ao);

    \draw[edge] (Co)  -- (C4a);
    \draw[edge] (C4a) -- (C4b);
    \draw[edge] (C4b) -- (C4c);
    \draw[edge] (C4c) -- (Co);

\end{tikzpicture}
    \caption{The oddomorphism obtained by gluing $C_4$ to the oddomorphism in Figure \ref{fig:oddo_example}}
    \label{fig:oddo_glued_example}
\end{figure}

We can now deduce the main technical lemma of this section. It shows that when looking for an example answering Question \ref{q} in the negative, we can impose additional constraints on the minor map. Namely, it needs to coincide with the oddomorphism on odd vertices that do not share a fibre with other odd vertices.
\begin{lemma}
\label{lemma:labelled-minors}
    Let $F$ and $G$ be connected graphs, $\psi: F \to G$ be an oddomorphism, $S$ be a set of vertices in $G$ such that the $\psi$-fibres of each element of $S$ contain exactly one $\psi$-odd vertex, and $\rho: S \to V(F)$ be the function that picks out this unique $\psi$-odd vertex.
    Then there exist graphs $F_+,G_+$ that extend $F,G$ and an oddomorphism $\psi_+: F_+\to G_+$ such that $G_+$ is a minor of $F_+$ if and only if there is a partial function $\pi: V(F) \rightharpoonup V(G)$ witnessing the fact that $G$ is a minor of $F$ which satisfies the condition that $\pi(\rho(s)) = s$ for each $s \in S$.
\end{lemma}
\begin{proof} 
    Observe that the graphs $F_+$ and $G_+$ from Lemma \ref{lemma:rooted-minors} are obtained by gluing star graphs to $F$ and $G$, respectively. It therefore follows from repeated application of Lemma \ref{lemma:oddo_extension} that an oddomorphism $F \to G$ can be extended to an oddomorphism from $F_+$ to $G_+$, where $F_+$ and $G_+$ are constructed as in Lemma \ref{lemma:rooted-minors} with $S = \{v_1,\ldots, v_n\}$ and $w_i = \rho(v_i)$. The result now follows from Lemma \ref{lemma:rooted-minors}.
\end{proof}

We now construct graphs $A$ and $B$ that will form the foundation of our example.
We define $B$ as the gluing of two cycles of length 5 at a single vertex (see Figure \ref{fig:G}). More specifically, it consists of vertices $\{c, v_1, v_2, v_3, v_4, w_1, w_2, w_3, w_4\}$ where $E(B) = \{ \{v_i, v_{i+1 \mod 5}\} \mid i \in \{0,\ldots, 4\}\} \cup \{ \{w_i, w_{i+1 \mod 5}\} \mid i \in \{0,\ldots, 4\}\}$ where $ v_0 = c = w_0$. 

\begin{figure}[htbp]
    \centering
    \begin{tikzpicture}[
    vertex/.style={circle, draw, fill=black, inner sep=2pt},
    vertex_name/.style={font=\scriptsize, inner sep=0pt, outer sep=0pt},
    thick
]

    \def\radius{1.5cm} %
    \def\labelRadius{1.9cm}

    \node[vertex] (center) at (0,0) {};
    \node[vertex_name, above=1mm] at (center.north) {$c$};

    \foreach \angle [count=\i] in {72, 144, 216, 288} {
        \node[vertex] (L\i) at ($(-\radius,0) + (\angle:\radius)$) {};
        \node[vertex_name] at ($(-\radius,0) + (\angle:\labelRadius)$) {$v_{\i}$};
    }
    
    \draw (center) -- (L1) -- (L2) -- (L3) -- (L4) -- (center);

    \foreach \angle [count=\i] in {252, 324, 36, 108} {
        \node[vertex] (R\i) at ($(\radius,0) + (\angle:\radius)$) {};
        \node[vertex_name] at ($(\radius,0) + (\angle:\labelRadius)$) {$w_{\i}$};
    }

    \draw (center) -- (R1) -- (R2) -- (R3) -- (R4) -- (center);

\end{tikzpicture}
    \caption{The graph $B$.}
    \label{fig:G}
\end{figure}
\par 
The graph $A$, along with the oddomorphism $\varphi:A \to B$, is constructed by making each fibre consist of a single odd and a single even vertex, and each bipartite graph between such fibres a path of length 3 starting and ending at the odd vertices (see Figure \ref{fig:F}). More specifically, for a vertex $u \in V(B)$ there are two vertices $(u,o), (u,e) \in V(A)$. If $\{u,t\} \in E(B)$ then $\{(u,o), (t,e)\}, \{(t,e), (u,e)\}, \{(u,e), (t,o)\} \in E(A)$. Then we can define $\varphi: (u,x) \mapsto u$ where $x \in \{o,e\}$.

\begin{figure}[htbp]
    \centering
    \begin{tikzpicture}[
    v_odd/.style={circle, draw=blue!80!black, fill=blue!20, inner sep=2pt, minimum size=6mm, font=\tiny},
    v_even/.style={rectangle, draw=red!80!black, fill=red!20, inner sep=2.5pt, minimum size=5mm, font=\tiny},
    vertex_name/.style={font=\scriptsize, inner sep=0pt, outer sep=0pt},
    edge/.style={draw=black!70, thin}
]

    \def\cycleRadiusOdd{1.8cm}  %
    \def\cycleRadiusEven{1.2cm} %
    \def\labelRadiusOdd{2.5cm}
    \def\shift{2.2cm}           %

    \newcommand{\connectbags}[4]{
        \draw[edge] (#1) -- (#4);
        \draw[edge] (#4) -- (#2);
        \draw[edge] (#2) -- (#3);
    }

    \node[v_odd] (C_o) at (0, 0.5) {};
    \node[v_even] (C_e) at (0, -0.5) {};
    \node[vertex_name, above=1mm] at (C_o.north) {$(c,o)$};
    \node[vertex_name, below=1mm] at (C_e.south) {$(c,e)$};

    \foreach \angle/\labelRadiusEven [count=\i] in {72/0.7cm, 144/0.5cm, 216/0.5cm, 288/0.7cm} {
        \coordinate (origin) at (-\shift,0);
        \node[v_odd]  (L\i_o) at ($(origin) + (\angle:\cycleRadiusOdd)$) {};
        \node[v_even] (L\i_e) at ($(origin) + (\angle:\cycleRadiusEven)$) {};
        \node[vertex_name] at ($(origin) + (\angle:\labelRadiusOdd)$) {$(v_{\i},o)$};
        \node[vertex_name] at ($(origin) + (\angle:\labelRadiusEven)$) {$(v_{\i},e)$};
    }

    \connectbags{C_o}{C_e}{L1_o}{L1_e}
    \connectbags{L1_o}{L1_e}{L2_o}{L2_e}
    \connectbags{L2_o}{L2_e}{L3_o}{L3_e}
    \connectbags{L3_o}{L3_e}{L4_o}{L4_e}
    \connectbags{L4_o}{L4_e}{C_o}{C_e}

    \foreach \angle/\labelRadiusEven [count=\i] in {252/0.7cm, 324/0.5cm, 36/0.5cm, 108/0.7cm} {
        \coordinate (origin) at (\shift,0);
        \node[v_odd]  (R\i_o) at ($(origin) + (\angle:\cycleRadiusOdd)$) {};
        \node[v_even] (R\i_e) at ($(origin) + (\angle:\cycleRadiusEven)$) {};
        \node[vertex_name] at ($(origin) + (\angle:\labelRadiusOdd)$) {$(w_{\i},o)$};
        \node[vertex_name] at ($(origin) + (\angle:\labelRadiusEven)$) {$(w_{\i},e)$};
    }

    \connectbags{C_o}{C_e}{R1_o}{R1_e}
    \connectbags{R1_o}{R1_e}{R2_o}{R2_e}
    \connectbags{R2_o}{R2_e}{R3_o}{R3_e}
    \connectbags{R3_o}{R3_e}{R4_o}{R4_e}
    \connectbags{R4_o}{R4_e}{C_o}{C_e}

\end{tikzpicture}
    \caption{Our graph $A$. The red squares are the even vertices and the blue circles are the odd vertices. The pairs of odd-even vertices that are next to each other in the figure form the fibres of the oddomorphism $\varphi:A \to B$.}
    \label{fig:F}
\end{figure}

The key property of these graphs is the following:
\begin{proposition}
\label{prop:labelled-counterexample}
    Let $A$, $B$, and $\varphi$ be defined as in the preceding text. There is no partial function $\pi: V(A) \rightharpoonup V(B)$ which witnesses the fact that $B$ is a minor of $A$ and also coincides with $\varphi$ on the $\varphi$-odd vertices of $A$.
\end{proposition}
\begin{proof}
    This is relatively simple to check by hand. Assume $\pi$ is such a partial function. To derive a contradiction, we examine two cases. For the remainder of this proof, a fibre is always a fibre of $\pi$.\begin{itemize}
        \item Assume $(c, e)$ is mapped to $c$ under $\pi$. Since $\pi^{-1}(c)$ is connected and it contains $(c,o)$ it must include another even vertex. Assume without loss of generality that it contains $(v_1, e)$.
        Now, $(v_1,o) \in \pi^{-1}(v_1)$ and this fibre must be adjacent to the fibre containing $(v_2,o)$, so $(v_2,e)$ must be contained in either $\pi^{-1}(v_1)$ or $\pi^{-1}(v_2)$.
        In both cases $(v_3,e)$ must be contained in $\pi^{-1}(v_2)$.
        But then $\pi^{-1}(v_4)$ can only consist of $(v_4,o)$ and is therefore not adjacent to $\pi^{-1}(v_3)$, a contradiction.

        \item Assume $(c, e)$ is not mapped to $c$ under $\pi$.
        Now, assume without loss of generality that $\pi((c,e)) = w_i$ for some $i \in \{1,2,3,4\}$. Since $\pi^{-1}(c)$ must be adjacent to both $\pi^{-1}(v_1)$ and $\pi^{-1}(v_4)$ it follows that $(v_1,e)$, $(v_2,e)$, $(v_3,e)$, and $(v_4,e)$ must all be contained in $\pi^{-1}(c) \cup \pi^{-1}(v_1) \cup \pi^{-1}(v_4)$ (since the only valid paths from $(c,o)$ to $(v_1,o)$ and $(v_4,o)$ go through these vertices). But then $\pi^{-1}(v_2)$ and $\pi^{-1}(v_3)$ are not adjacent, a contradiction.
    \end{itemize}

\end{proof}

It now follows from Lemma \ref{lemma:labelled-minors} that we can construct an oddomorphism $\varphi_+: A_+ \to B_+$ where $B_+$ is not a minor of $A_+$. We have thus resolved Question \ref{q} by showing:
\begin{theorem}
    There exists a pair of connected graphs $F$, $G$ such that there is an oddomorphism $F \to G$ but $G$ is not a minor of $F$.
\end{theorem}

This also shows that the class of graphs obtained by excluding $G = B_+$ as a minor is not closed under oddomorphisms while still being minor closed.
Additionally, $B$ is planar and $B_+$ is obtained by adding leaves to $B$, so $B_+$ is also planar. It then follows from a celebrated result by Robertson and Seymour \cite{ROBERTSON198692} that this class of graphs that exclude $B_+$ as a minor has bounded treewidth.
Thus the statement that a minor closed class is closed under oddomorphisms fails even for classes of bounded treewidth.
This observation was made by Tim Seppelt (personal communication).

\subsection{Second Example: Random Lifts of Cliques}
The second example follows from work by Drier and Linial on the Hadwiger number of random lifts of complete graphs \cite{Drier2004Minors}. \par 
For a positive integer $l$, an \emph{$l$-lift} of a graph $G$ is a graph with vertex set $V(G) \times [l]$ and an edge set consisting of perfect matchings between $\{u\} \times [l]$ and $\{v\} \times [l]$ for each $\{u,v\} \in E(G)$. We say that $F$ is a \emph{lift} of $G$ if it is an $l$-lift of $G$ for some positive integer $l$.
\par 
The key observation explaining the relevance of this work for our purposes is that odd lifts are examples of oddomorphisms.
\begin{proposition}[Lemma 3.18 in~\cite{roberson2022oddomorphisms}]
    Let $G_l$ be an $l$-lift of a graph $G$ where $l$ is an odd number. Then the projection $G_l \to G, (v,i) \mapsto v$ is an oddomorphism.
\end{proposition}
\begin{proof}
    The projection is clearly a homomorphism.
    We let every vertex of $G_l$ be odd. Since $l$ is odd, this means that there is an odd number of odd vertices in each fibre of the projection. The inverse image of an edge in $G$ is, by definition, a perfect matching in $G_l$. Therefore each vertex has degree 1 in that subgraph and thus they all satisfy the degree condition.
\end{proof}

A \emph{random $l$-lift} of a graph $G$ is an $l$-lift of $G$ such that for each $\{u,v\} \in E(G)$ the perfect matching between $\{u\} \times [l]$ and $\{v\} \times [l]$ is chosen uniformly at random.

\begin{lemma}[Consequence of Lemma 3.8 in \cite{Drier2004Minors}]
\label{lemma:Drier}
    For every integer $l\geq 2$ and every $\delta>0$, there exists an integer $N$ such that for every $n\geq N$, the probability that a random $l$-lift of $K_n$ contains $K_n$ as a minor is less than $\delta$.
\end{lemma}

In particular, by Lemma~\ref{lemma:Drier}, the probability that a random $3$-lift of $K_n$ contains $K_n$ as a minor tends to zero as $n\to\infty$.
Since these lifts have an oddomorphism to $K_n$ for odd $l$, this gives an infinite class of examples confirming the negative answer to Question \ref{q}.

\section{The Strong Roberson Conjecture for Split-Off Minors Is True}
\label{sec:oddo-means-som}
In the previous section we showed that the existence of an oddomorphism from $F$ to $G$ is not enough to ensure that $F$ contains $G$ as a minor. In this section we show that a slightly weaker relation between $F$ and $G$ holds. Namely, that $F$ contains $G$ as a split-off minor.

\begin{definition}
\label{def:split-off-minors}
    A multigraph $H$ is a \emph{split-off minor} of a multigraph $G$ if $H$ can be obtained from $G$ by applying the operations of taking subgraphs, contracting edges, and splitting off edges (both considered as multigraph operations).
\end{definition}

In particular, if $G$ and $H$ are simple graphs, we say that $G$ contains $H$ as a split-off minor if $H$ can be obtained from $G$ by applying these operations as multigraph operations. Thus the graphs in the intermediate steps can have multiple edges.
\par 
The split-off minor relation is a natural one to consider when studying classes closed under such operations. Indeed, a class of multigraphs is closed under taking split-off minors if and only if it is closed under both minors and immersions. See Proposition \ref{prop:split-off-closure} for further details.

\par 
The split-off minor relation has not been studied previously in the literature. 
Golovach, Paulusma, Kamiński, and Thilikos~\cite{GOLOVACH2014286} examined a very similar relation which they call the \emph{lift-minor} relation.
Their relation is defined using the same operations. However, they restrict themselves only to simple graphs. Additionally, they consider the operations as simple graph operations, so multiple edges are always merged into single edges when they arise and loops are always deleted.
\par

The goal of this section is to prove Theorem \ref{thm:oddo-implies-lm}.
We begin by proving it in the case where $F$ is a lift of $G$.

\begin{lemma}
\label{lemma:perfect-matching-lm}
    If $F$ is a lift of a graph $G$ then $G$ is a split-off minor of $F$.
\end{lemma}
\begin{proof}
    Let $\varphi:F\to G$ be the projection and let $T$ be a full spanning forest in $G$, meaning that each component of it is a spanning tree of a component of $G$. 
    An $l$-lift of a forest consists of $l$ isomorphic copies of that forest. We therefore see that there exists a subforest $T_F$ in $F$ such that $\varphi$ restricts to an isomorphism $T_F \to T$.
    \begin{claim}
    \label{claim:tree-paths}
        For each $x,y \in T_F$ that lie on the same tree in $T_F$ and satisfy $\{\varphi(x), \varphi(y)\} \in E(G) \setminus E(T)$ there exists a path $P_{x,y}$ between $x$ and $y$ that is disjoint from $T_F$ except at its endpoints, and additionally the first and last edge of the path both map to $\{\varphi(x), \varphi(y)\}$ under $\varphi$.
    \end{claim}
    \begin{proofofclaim}
        Assume $x$ and $y$ are as stated above. Note that there is a cycle $C$ in $G$ consisting only of the edge $\{\varphi(x), \varphi(y)\}$ and edges from $T$.
        The graph $\varphi^{-1}(C)$, with vertex set $\varphi^{-1}(V(C))$ and edge set $\varphi^{-1}(E(C))$, is a lift of $C$, so every vertex in it has degree 2 and it is therefore a disjoint union of cycles.
        Note that the path between $x$ and $y$ which lies in the tree $T_F$ is contained in this subgraph and thus $x$ and $y$ lie on the same cycle $C'$ within $\varphi^{-1}(C)$. Taking the other path between $x$ and $y$ on that cycle gives a path between $x$ and $y$ that is disjoint from $T_F$ except at its endpoints.
        The first and last edges of this path do not map to $T$ under $\varphi$, since if that were the case $x$ or $y$ would have degree larger than 1 within the preimage of the edge mapped to, contradicting the fact that this preimage is a perfect matching. Thus these first and last edges must map to the only edge of $C$ not in $T$, which is $\{\varphi(x), \varphi(y)\}$.
    \end{proofofclaim}

    We now form a graph $F'$ from $F$ by contracting a spanning tree of each component of $F\setminus V(T_F)$ and deleting the resulting loops.
    Note that $F'$ is not necessarily a simple graph.
    We prove that $G$ immerses into $F'$. For the branch vertex function we use the inverse of $\varphi$ on $T_F$, that is $(\restr{\varphi}{T_F})^{-1}$. For $e \in E(T)$ we define $P_e$ as the path consisting only of the corresponding edge in $T_F$. For edges $e$ outside of $T$, we define $P_e$ as the path in $F'$ corresponding to the path obtained from Claim \ref{claim:tree-paths}. For each path obtained from the claim, its internal vertices all lie in the same component of $F \setminus V(T_F)$ and are thus identified by the contractions. Thus the paths $P_e$ with $e \in E(G) \setminus E(T)$ consist only of the first and last edges of the corresponding path obtained from Claim~\ref{claim:tree-paths}. From the last statement of Claim~\ref{claim:tree-paths} it then follows that the edges of such $P_e$ come from edges of $F$ that are mapped to $e$ under $\varphi$. By definition this also holds for $P_e$ with $e \in E(T)$. Thus the paths $P_e$ are pairwise edge-disjoint, which completes the proof.
\end{proof}

We can now prove the full theorem.

\begin{proof}[Proof of Theorem \ref{thm:oddo-implies-lm}]
    Let $\varphi:F\to G$ be an oddomorphism, where $F$ may be a multigraph.
    We can assume that $G$ is connected since the oddomorphism induces an oddomorphism to each component. Additionally, the result is trivial if $G$ is the single vertex graph, so we assume $G$ has at least one edge and thus has no isolated vertices.
    We show that if $F$ is not a lift of $G$, then there is a proper split-off minor $F'$ of $F$ with an oddomorphism to $G$.
    
    Suppose that $F$ is not a lift of $G$. If every edge preimage induced a perfect matching, connectedness of $G$ would force all fibres to have the same size, making $F$ a lift of $G$. Thus there is an edge $e\in E(G)$ such that the subgraph of $F$ induced by $\varphi^{-1}(e)$ is not a perfect matching. Choose a vertex $v$ whose degree in this subgraph is not 1. If $v$ is isolated in $F$, then deleting $v$ gives a proper split-off minor $F'$ of $F$ for which $\restr{\varphi}{F'}:F'\to G$ is still an oddomorphism.
    If $v$ has degree 0 in this induced subgraph but is not isolated in $F$, then it is $\varphi$-even. Thus there is an edge $e'\in E(G)$ incident with $\varphi(v)$ such that $v$ has degree at least 2 in the subgraph induced by $\varphi^{-1}(e')$.

    In either remaining case, there is an edge $f\in E(G)$ incident with $\varphi(v)$ such that $v$ has degree at least 2 in the subgraph induced by $\varphi^{-1}(f)$. Choose two distinct edges incident with $v$ in this subgraph, and let $u,w$ be their other endpoints. Then $\varphi(u)=\varphi(w)$.
    If $u=w$, delete these two parallel edges to obtain $F'$. This preserves all degree parities, so the same map $\varphi:F'\to G$ is still an oddomorphism.
    Suppose now that $u\neq w$. Split off $\{u,v\}$ and $\{v,w\}$, then contract the resulting edge $\{u,w\}$ to obtain $F'$. This operation is illustrated in Figure~\ref{fig:split-off-contract}.

    \begin{figure}[htbp]
        \centering
        \begin{tikzpicture}[
    vertex/.style={circle, draw=black, thick, minimum size=14pt,
                   inner sep=0pt, font=\small},
    odd/.style={vertex, fill=blue!10},
    even/.style={vertex, fill=white},
    edge/.style={thick, black},
    selected/.style={very thick, blue!65!black},
    context/.style={thick, gray},
    bag/.style={ellipse, draw=gray, dashed, inner sep=6pt},
    every label/.style={font=\small},
    >=stealth
]
    \begin{scope}
        \node[odd, label=above:{$u$}]  (u0) at (0, 0.5) {o};
        \node[even, label=below:{$w$}] (w0) at (0,-0.5) {e};
        \node[odd, label=above:{$v$}] (v0) at (1.55, 0.75) {o};
        \node[even]                  (y0) at (1.55,-0.75) {e};
        \node[bag, fit=(u0) (w0)] (A0) {};
        \node[bag, fit=(v0) (y0)] (B0) {};

        \draw[context] (u0) -- ++(-0.95,0);
        \draw[context] (w0) -- ++(-0.95, 0.18);
        \draw[context] (w0) -- ++(-0.95,-0.18);
        \draw[context] (v0) -- ++(0.90,0);
        \draw[context] (y0) -- ++(0.90, 0.18);
        \draw[context] (y0) -- ++(0.90,-0.18);
        \draw[edge] (u0) to[bend left=18] (v0);
        \draw[edge] (u0) -- (y0) -- (w0);
        \draw[selected] (u0) -- (v0) -- (w0);

        \node[font=\small] at (0.775,2.00) {$F[\varphi^{-1}(f)]$};
        \node[font=\small] at (0,-1.85) {$a$};
        \node[font=\small] at (1.55,-1.85) {$b$};
    \end{scope}

    \draw[->, thick] (2.60,0) -- (3.70,0)
        node[midway, above=3pt, font=\small] {split off}
        node[midway, below=3pt, font=\small] {at $v$};

    \begin{scope}[xshift=4.9cm]
        \node[odd, label=above:{$u$}]  (u1) at (0, 0.5) {o};
        \node[even, label=below:{$w$}] (w1) at (0,-0.5) {e};
        \node[odd, label=above:{$v$}] (v1) at (1.55, 0.75) {o};
        \node[even]                  (y1) at (1.55,-0.75) {e};
        \node[bag, fit=(u1) (w1)] (A1) {};
        \node[bag, fit=(v1) (y1)] (B1) {};

        \draw[context] (u1) -- ++(-0.95,0);
        \draw[context] (w1) -- ++(-0.95, 0.18);
        \draw[context] (w1) -- ++(-0.95,-0.18);
        \draw[context] (v1) -- ++(0.90,0);
        \draw[context] (y1) -- ++(0.90, 0.18);
        \draw[context] (y1) -- ++(0.90,-0.18);
        \draw[edge] (u1) to[bend left=18] (v1);
        \draw[edge] (u1) -- (y1) -- (w1);
        \draw[selected] (u1) -- (w1);

        \node[font=\small] at (0.775,2.00) {Intermediate graph};
        \node[font=\small] at (0,-1.85) {$a$};
        \node[font=\small] at (1.55,-1.85) {$b$};
    \end{scope}

    \draw[->, thick] (7.50,0) -- (8.60,0)
        node[midway, above=3pt, font=\small] {contract}
        node[midway, below=3pt, font=\small] {$uw$};

    \begin{scope}[xshift=9.8cm]
        \node[odd]                         (uw) at (0,0) {o};
        \node[odd, label=above:{$v$}] (v2) at (1.55, 0.75) {o};
        \node[even]                  (y2) at (1.55,-0.75) {e};
        \node[bag, fit=(uw)] (A2) {};
        \node[bag, fit=(v2) (y2)] (B2) {};
        \node[font=\small, above=5pt of A2.north] {$\{u,w\}$};

        \draw[context] (uw) -- ++(-0.95, 0.30);
        \draw[context] (uw) -- ++(-0.95,0);
        \draw[context] (uw) -- ++(-0.95,-0.30);
        \draw[context] (v2) -- ++(0.90,0);
        \draw[context] (y2) -- ++(0.90, 0.18);
        \draw[context] (y2) -- ++(0.90,-0.18);
        \draw[edge] (uw) -- (v2);
        \draw[edge] (uw) to[bend left=12] (y2);
        \draw[edge] (uw) to[bend right=12] (y2);

        \node[font=\small] at (0.775,2.00) {$F'[\varphi'^{-1}(f)]$};
        \node[font=\small] at (0,-1.85) {$a$};
        \node[font=\small] at (1.55,-1.85) {$b$};
    \end{scope}
\end{tikzpicture}
        \caption{The split-off and contraction step. Here, $f=\{a,b\}$.}
        \label{fig:split-off-contract}
    \end{figure}

    Let $\varphi':F'\to G$ send $\{u,w\}$ to $\varphi(u)=\varphi(w)$ and agree with $\varphi$ on every other vertex. This is a homomorphism.
    For every vertex $z\in V(F')$ other than $\{u,w\}$ and every edge $g\in E(G)$ incident with $\varphi(z)$, compare the degree of $z$ in the subgraph of $F'$ induced by $\varphi'^{-1}(g)$ with its degree in the subgraph of $F$ induced by $\varphi^{-1}(g)$. These degrees are equal unless $(z,g)=(v,f)$, in which case the former is smaller by 2. Their parities are therefore equal, so these vertices can retain their odd/even labels and satisfy condition (ii) of Definition~\ref{def:oddomorphisms}.

    Assign the contracted vertex the parity of the sum of the parities of $u,w$ under $\varphi$. Replacing $u,w$ by this vertex preserves the parity of the number of odd vertices in their fibre, so condition (i) holds for $\varphi'$.
    Finally, let $g\in E(G)$ be incident with $\varphi'(\{u,w\})$. The degree of $\{u,w\}$ in the subgraph of $F'$ induced by $\varphi'^{-1}(g)$ is the sum of the degrees of $u,w$ in the subgraph of $F$ induced by $\varphi^{-1}(g)$, with 2 subtracted if $g=f$. This degree has the same parity as the sum of the parities of $u,w$, as required by condition (ii). Thus $\varphi':F'\to G$ is an oddomorphism.

    The sum of the numbers of edges and vertices is strictly smaller in $F'$ than in $F$. When we iterate this construction, we therefore arrive, in finitely many steps, at a split-off minor $F^*$ of $F$ that is a lift of $G$. Lemma~\ref{lemma:perfect-matching-lm} now implies that $G$ is a split-off minor of $F^*$, and hence of $F$.
    
\end{proof}

The following corollary follows from Theorem \ref{thm:oddo-implies-lm} and the fact that a subgraph of a graph $F$ is a split-off minor of $F$.
\begin{corollary}
    \label{cor:wo-implies-som}
    If a graph $F$ has a weak oddomorphism to a graph $G$ then $G$ is a split-off minor of $F$.
    Therefore, a class of graphs that is closed under taking split-off minors is closed under weak oddomorphisms.
\end{corollary}

\section{Split-Off Minors}
The goal of this section is to establish some facts about the split-off minor relation. 
\par 
We begin by observing that it forms a \emph{well-quasi-ordering}.
One of the most celebrated results in structural graph theory is the result of Robertson and Seymour that the graph minor relation is a well-quasi-ordering on the set of all finite graphs \cite{ROBERTSON2004325}. A well-quasi-ordering is a quasi-order (also known as preorder) that has neither infinite descending chains nor infinite antichains.
A consequence of the minor relation being a well-quasi-order is that for every class $\C$ of graphs that is closed under taking minors there exists a finite set of graphs $S$ such that $\C$ is the class of graphs that exclude the graphs in $S$ as minors.
The split-off minor relation inherits this property from the minor relation.
\begin{proposition}
\label{prop:split-off-is-wqo}
    The split-off minor relation is a well-quasi-ordering on the set of finite graphs.
\end{proposition}
\begin{proof}
    Since the number of edges or vertices is always reduced when any of the operations are applied, it follows that there are no infinite descending chains. Since a minor is a split-off minor and there are no infinite antichains in the minor relation, it follows that there are no infinite antichains in the split-off minor relation. 
\end{proof}
Note that Robertson and Seymour also proved that the immersion relation is a well-quasi-ordering \cite{ROBERTSON2010181}. We could therefore also have said that the split-off minor relation inherits the well-quasi-ordering property from the immersion relation.
\par
The graph minor and immersion relations have both \emph{constructive} and \emph{structural} definitions. The former is the definition that states that $G$ is a minor/immersion of $F$ if it can be constructed from $F$ by applying a sequence of operations. The latter is the static definition, which is based on whether the structure of the minor/immersion can be exhibited within the larger graph in some way.
\par 
Given this distinction, the definition given here for split-off minors is constructive.
This evokes the question of whether there is also a structural definition for split-off minors. 
It proved difficult to find such a characterisation that is simple, so we settle for one that is quite involved but nevertheless insightful.
Here, we use $\uplus$ to denote the additive union of two multisets and $\setminus$ to denote their subtraction.
\begin{theorem}
\label{thm:split-off-structural}
    The following are equivalent:
    \begin{enumerate}[(i)]
        \item $G$ is a split-off minor of $F$.
        \item There exists a forest $T$ on $V(F)$, a linear ordering $\prec$ on the edges of $T$, a set of paths $\{P_t \mid t \in E(T)\}$, an injective function $\pi$ from $V(G)$ to the components of $T = (V(F), E(T))$, and a set of paths $\{P_e \mid e \in E(G)\}$ such that \begin{itemize}
            \item For $t \in E(T)$, $P_t$ is a path between the endpoints of $t$ consisting of edges from $(E(F) \setminus \biguplus_{t'\prec t} P_{t'}) \uplus \{t' \mid t'\prec t\}$.
            
            \item $P_{\{u,v\}}$ for $\{u,v\} \in E(G)$ is a path between a vertex in $\pi(u)$ and a vertex in $\pi(v)$, consisting of edges from $(E(F) \setminus \biguplus_{t\in E(T)} P_{t}) \uplus E(T)$.
            \item For any distinct $e,e' \in E(G)$, we have $E(P_e) \cap E(P_{e'}) \subseteq E(T)$.
        \end{itemize}
    \end{enumerate}
\end{theorem}
The reason for the necessity of the linear order $\prec$ in the above characterisation is that the order in which the splitting-off and the contraction operations are applied matters. When a splitting-off operation is applied to create an edge that is then contracted, the contraction cannot be applied before the splitting off. Likewise, if a path that is to be split off takes advantage of the fact that two vertices have been merged, the splitting off cannot happen before the contraction (without removing the edge that was contracted to merge the two vertices). It therefore seems necessary to retain some dependency order in the structural definition of split-off minors.

\begin{proof}\ %
    \begin{description}
        \item[(ii)$\Rightarrow$(i)]
            We iterate through the edges $t \in E(T)$ in the order of $\prec$ and first split off the path $P_t$ and then contract the resulting edge $t$. Note that for $t'\prec t$, the edge $t'$ has already been contracted, so, to be precise, we do not split off $P_t$ but the path obtained from $P_t$ by contracting all of the edges $\{t' \mid t'\prec t\}$. Let $F'$ be the resulting graph.
            \par
            The vertices of $F'$ correspond to the connected components of $T$. Its edges correspond to the edges in $E(F) \setminus \biguplus_{t\in E(T)} P_t$ between the corresponding components of $T$.
            We therefore see that $\pi$ and the paths obtained by contracting all edges in $T$ from $\{P_e \mid e \in E(G)\}$ define an immersion of $G$ into $F'$.
            This shows that $G$ is a split-off minor of $F$.

        \item[(i)$\Rightarrow$(ii)]
        We define the required objects inductively.
        First note that if $G = F$, then we can define $T = (V(F), \varnothing)$, $\pi = \id_{V(F)}$, and for $e \in E(G) = E(F)$ we let $P_e$ be the path consisting only of $e$. It is easy to see that all of the conditions are satisfied.\par
        For the inductive case, let $G'$ be a graph that is a split-off minor of $F$ such that $G$ can be obtained from $G'$ by applying exactly one of the allowed operations.
        By the induction hypothesis we assume that there exist $T_{G'}$, $\prec_{G'}$, $\{P_t \mid t \in E(T_{G'})\}$, $\pi_{G'}$, and $\{P_e \mid e \in E(G')\}$ satisfying the conditions.
        We show how the corresponding objects $T_{G}$, $\prec_G$, $\{P_t \mid t \in E(T_{G})\}$, and $\pi_{G}$ are constructed in the four different cases.
        \begin{itemize}
            \item $G$ is obtained from $G'$ by deleting an edge.
            The only change needed is to delete the path $P_e$ corresponding to the deleted edge $e \in E(G')$.
            \item $G$ is obtained from $G'$ by deleting an isolated vertex. In this case we only have to remove the isolated vertex from the domain of $\pi_{G'}$.

            \item $G$ is obtained from $G'$ by splitting off the edges $\{v,u\}, \{u,w\} \in E(G')$. In this case we remove the paths $P_{\{v,u\}}$ and $P_{\{u,w\}}$ and add $P_{\{v,w\}}$, which is the composition of $P_{\{v,u\}}$, a path within the component of $\pi(u)$ in $T_G$, and $P_{\{u,w\}}$.

            \item $G$ is obtained from $G'$ by contracting an edge $e' \in E(G')$. The set $\{P_e \mid e \in E(G)\}$ is inherited from $G'$ by omitting $P_{e'}$. Let $u,v \in V(F)$ be the endpoints of the path $P_{e'}$. We add a new edge $t$ between $u$ and $v$ to $T_{G'}$ to form $T_G$ and put $t$ at the end of the order $\prec_{G'}$ to form the order $\prec_G$. We then define $P_t$ as the path $P_{e'}$.
            The function $\pi_G$ is equal to $\pi_{G'}$ at all vertices except at the newly merged one. It maps that vertex to the component of $T_{G}$ obtained by joining the components of $T_{G'}$ that contain $u$ and $v$.
        \end{itemize}
         It is easy to check that the constructions $T_{G}$, $\prec_G$, $\{P_t \mid t \in E(T_{G})\}$, and $\pi_{G}$ satisfy the conditions in all of the cases.
    \end{description}
\end{proof}

As mentioned previously, the split-off minor relation is very similar to the lift-minor relation studied in~\cite{GOLOVACH2014286}, the only difference being that the multiplicity of edges is always immediately reduced to 1 when forming a lift-minor. It is easy to see that if a graph $G$ is a lift-minor of a graph $F$, then $G$ is a split-off minor of $F$.
Thus all positive results about the existence of lift-minors also hold for split-off minors.
In particular, the closure of a class of unbounded treewidth under taking split-off minors is always the class of all graphs~\cite[Theorem~1.1]{GOLOVACH2014286}, which implies that every proper class closed under taking split-off minors has bounded treewidth.

It is not immediately clear whether the lift-minor and the split-off minor relations are equivalent.
When taking graph minors, whether multiple edges are deleted immediately does not matter, provided that both the initial and final graphs are simple.
An analogous statement holds for immersions, but its proof is not quite immediate.
Since the author was unable to locate a reference for this fact in the literature, the result and its proof are included here.
\begin{proposition}
\label{prop:simple-immersion}
    Let $G,F$ be simple graphs such that $G$ immerses into $F$. Then $G$ can be obtained from $F$ using the operations of splitting off edges and taking a subgraph in such a way that at every intermediate step the graph is simple.
\end{proposition}
\begin{proof}
    Let $\imm$ denote the immersion relation and let $\immsimp$ denote the relation corresponding to the existence of an immersion where all intermediate graphs are simple.
    Let $(\pi, \{P_e \mid e \in E(G)\})$ be the immersion of $G$ into $F$.
    We show that we can find a simple graph $F'$ with fewer edges than $F$ such that $G \imm F' \immsimp F$. The result then follows from an induction on the number of edges in $F$.
    To find $F'$ we examine a few cases:
    \begin{itemize}
        \item $F$ contains an edge that does not appear in any $P_e$. We can then let $F'$ be the subgraph of $F$ obtained by deleting that edge. It is clear that $G \imm F' \immsimp F$ and $F'$ has fewer edges than $F$.

        \item $F$ only has edges that appear in some $P_e$ and additionally there exist vertices $v,u,w$ that lie on a path $P_e$ in this order such that $\{v,u\}, \{u,w\} \in E(F)$ but $\{v,w\} \notin E(F)$.
        We then split off the edges $\{v,u\}$ and $\{u,w\}$ to form the graph $F'\immsimp F$. We now show that $G$ immerses into $F'$. To see this, let $P_{e'}$ be the path containing $\{v,u\}$ and let $P_{e''}$ be the path containing $\{u,w\}$. To define the immersion, we leave $\pi$ unchanged along with all $P_f$ for $f \notin \{e,e', e''\}$. The path $P_e$ now takes the new shortcut $\{v,w\}$, the path $P_{e'}$ (if $e' \neq e$) is routed through the old route of $P_e$ between $v$ and $u$, and the path $P_{e''}$ (if $e'' \neq e$) is routed through the old route of $P_e$ between $u$ and $w$.

        \item Neither of the above cases holds. We then have that for each path $P_e$, $F$ contains the complete graph on the vertices of $P_e$. But then $G$ is a subgraph of $F$. Namely, it follows that for each $\{v,u\} \in E(G)$ we have $\{\pi(v), \pi(u)\} \in E(F)$, so $G \imm G \immsimp F$ and we are done.
    \end{itemize}
\end{proof}

In light of these two facts, it is slightly surprising that when the operations of edge-contraction and splitting off are allowed in conjunction, we can no longer require the intermediate steps to be simple and thus the split-off minor relation is not equal to the lift-minor relation.

\begin{proposition}
\label{prop:lift-minor-separation}
    There exist simple graphs $F$ and $G$ such that $G$ is a split-off minor of $F$ but $G$ is not a lift-minor of $F$.
\end{proposition}
\begin{proof}[Proof sketch]
    A result similar to Lemma \ref{lemma:rooted-minors} can be proved for lift-minors and split-off minors. Thus it suffices to find examples of labelled simple graphs $A$ and $B$ such that $B$ is a split-off minor but not a lift-minor of $A$.
    \begin{figure}[htbp]
        \centering
        \begin{minipage}{0.48\textwidth}
    \centering
    \begin{tikzpicture}[scale=0.8, thick]
        \node[circle, draw, inner sep=2pt, label=above:{$a_1$}] (a1) at (0, 2) {};
        \node[circle, draw, inner sep=2pt]                      (u1) at (1.5, 0) {};
        \node[circle, draw, inner sep=2pt, label=above:{$a_2$}] (a2) at (2.5, 2) {};
        \node[circle, draw, inner sep=2pt, label=below:{$a_5$}] (a5) at (3.5, 0) {};
        \node[circle, draw, inner sep=2pt, label=above:{$a_3$}] (a3) at (4.5, 2) {};
        \node[circle, draw, inner sep=2pt]                      (u2) at (5.5, 0) {};
        \node[circle, draw, inner sep=2pt, label=above:{$a_4$}] (a4) at (7, 2) {};

        \draw (a1) -- (u1);
        \draw (u1) -- (a2);
        \draw (u1) -- (a5);
        \draw (a2) -- (a5);
        \draw (a5) -- (a3);
        \draw (a5) -- (u2);
        \draw (a3) -- (u2);
        \draw (u2) -- (a4);
    \end{tikzpicture}
\end{minipage}\hfill
\begin{minipage}{0.48\textwidth}
    \centering
    \begin{tikzpicture}[scale=0.8, thick]
        \node[circle, draw, inner sep=2pt, label=above:{$a_1$}] (a1) at (0, 0) {};
        \node[circle, draw, inner sep=2pt, label=above:{$a_3$}] (a3) at (1.5, 0) {};
        \node[circle, draw, inner sep=2pt, label=below:{$a_5$}] (a5) at (3, 0) {};
        \node[circle, draw, inner sep=2pt, label=above:{$a_2$}] (a2) at (4.5, 0) {};
        \node[circle, draw, inner sep=2pt, label=above:{$a_4$}] (a4) at (6, 0) {};

        \draw (a1) -- (a3);
        \draw (a3) -- (a5);
        \draw (a5) -- (a2);
        \draw (a2) -- (a4);
    \end{tikzpicture}
\end{minipage}
        \caption{The graph $A$ (left) and $B$ (right).}
        \label{fig:lift-minor-counterexample}
    \end{figure}

    Define $A$ and $B$ and the vertices $\{a_1,\ldots, a_5\}$ as in Figure \ref{fig:lift-minor-counterexample}.
    To obtain $B$ as a split-off minor of $A$ we contract the two edges from $a_5$ to the unlabelled vertices and then split off the paths of length two from $a_1$ to $a_3$, and from $a_2$ to $a_4$.
    \par
    To see that $B$ is not a lift-minor of $A$ we use brute-force case analysis. This can be checked more easily by a computer program, but we include a check by hand here for completeness.\par
    We show that if any edge is deleted, contracted, or a pair of edges is split off in $A$ without creating multiple edges, a graph is obtained that does not contain $B$ as a split-off minor.
    A key observation we will use repeatedly is that if a vertex $w$ lies on a path $P$ between $u$ and $v$, it will stay on that path unless the two edges of $P$ that are incident to $w$ are split off.

    \begin{figure}[htbp]
        \centering
        \tikzset{
    vertex/.style={
        circle,
        draw,
        minimum size=4pt, %
        inner sep=0pt
    },
    label_node/.style={
        draw=none,
        inner sep=1pt,
        anchor=center,
        font=\small %
    },
    c_label/.style={
        label_node,
        above,
        yshift=0.05cm %
    }
}

\begin{tikzpicture}[scale=1.5]
    \node[vertex] (c1) at (0, 1.5) [label={[yshift=0.05cm]above:$a_1$}] {};
    \node[vertex] (v1) at (1, 0) {};
    \node[vertex] (c2) at (2, 1.5) [label={[yshift=0.05cm]above:$a_2$}] {};
    \node[vertex] (c5) at (3, 0) [label={[yshift=-0.05cm]below:$a_5$}] {}; %
    \node[vertex] (c3) at (4, 1.5) [label={[yshift=0.05cm]above:$a_3$}] {};
    \node[vertex] (v2) at (5, 0) {};
    \node[vertex] (c4) at (6, 1.5) [label={[yshift=0.05cm]above:$a_4$}] {};

    \draw (c1) -- (v1);
    \draw (v1) -- (c2);
    \draw (v1) -- (c5);
    \draw (c2) -- (c5);
    \draw (c5) -- (c3);
    \draw (c5) -- (v2);
    \draw (c3) -- (v2);
    \draw (v2) -- (c4);

    \node[label_node] at ($(c1)!0.5!(v1) + (0.05, 0.15)$) {$x$};
    \node[label_node] at ($(v1)!0.5!(c2) + (-0.05, 0.15)$) {$y$};
    \node[label_node] at ($(v1)!0.5!(c5) + (0, -0.15)$) {$r$};
    \node[label_node] at ($(c2)!0.5!(c5) + (0.05, 0.15)$) {$z$};
    \node[label_node] at ($(c5)!0.5!(c3) + (-0.05, 0.15)$) {$z'$};
    \node[label_node] at ($(c5)!0.5!(v2) + (0, -0.15)$) {$r'$};
\end{tikzpicture}
        \caption{The graph $A$ with edge labels.}
        \label{fig:lift-minor-counterexample-with-labels}
    \end{figure}
    For ease of reference, we label the edges as in Figure \ref{fig:lift-minor-counterexample-with-labels}.
    \begin{itemize}
        \item Edge deletion. Up to symmetry, we need to explore four cases.
        \begin{itemize}
            \item Delete $x$. Then $a_1$ is isolated, which is a property preserved by split-off minors, so the resulting graph cannot contain $B$ as a split-off minor.
            \item Delete $y$. Then $a_2$ has degree 1 and its degree cannot be increased by merging it with other vertices without merging it with $a_5$, which is forbidden. This graph therefore cannot contain $B$ as a split-off minor.

            \item Delete $r$. Then every path between $a_1$ and $a_3$ goes through $a_2$. Again, $a_2$ cannot be merged with other vertices to increase its degree. Thus $a_2$ will be isolated after splitting off the path around it.

           \item Delete $z$. The argument is similar to the previous cases.
        \end{itemize}
        \item Splitting off edges. We only need to explore cases that do not create multiple edges.
        \begin{itemize}
            \item Split off $(x,y)$. As before, $a_2$ lies on the only path between $a_1$ and $a_3$ and it cannot be merged to increase its degree.
            \item Split off $(x,r)$. This creates a leaf at the other endpoint of $y$. This leaf can clearly not be used to find the split-off minor and in the rest of the graph $a_2$ only has degree 1 and cannot be merged to another vertex to increase its degree enough.
            \item Split off $(r,z')$. The two paths between $a_2$ and $a_4$ contain $a_3$ and $a_5$ respectively. Neither of them can be split off since they have degree 2 and cannot be merged to increase their degree.

            \item The cases for splitting off $(r,r')$ and $(z,z')$ are similar to the previous case.
        \end{itemize}

        \item Contracting edges. The only contraction (up to symmetry) that does not produce multiple edges is the contraction of $x$. Note that the next (non-symmetric) operation must be an edge deletion or splitting off. All such operations give, by the previous cases, a graph that does not contain $B$ as a split-off minor.
    \end{itemize}
\end{proof}

\section{Applications in Finding h.d.\ Closed Classes}

In this section we use Theorem \ref{thm:oddo-implies-lm} to find homomorphism distinguishing closed classes of graphs. 
We begin by stating a convenient corollary that follows directly from Theorem \ref{thm:oddo-implies-hdclosed} and Corollary \ref{cor:wo-implies-som}.
\begin{corollary}
\label{cor:som-means-hdclosed}
    Every class of graphs that is closed under taking split-off minors and disjoint unions is h.d.\ closed.
\end{corollary}

We can restate this by observing that if $G$ is a connected multigraph then a graph $F$ contains $G$ as a split-off minor if and only if there is a connected component of $F$ that contains $G$ as a split-off minor. Therefore the set of graphs that exclude a connected multigraph as a split-off minor is closed under disjoint unions.
\begin{corollary}
\label{cor:thm6wqo-variant}
    For every finite set $S$ of connected multigraphs, the class of graphs that exclude the elements of $S$ as split-off minors is h.d.\ closed.
\end{corollary}
It is enough to consider finite sets $S$ since the split-off minor relation is a well-quasi-ordering, as shown in Proposition~\ref{prop:split-off-is-wqo}.
This of course gives rise to countless new examples of h.d.\ closed classes. However, it is hard to get a grasp on what properties these classes have.
\par

To identify classes that are defined by other properties but are also closed under taking split-off minors, we make use of the following characterisation.

\begin{proposition}
\label{prop:split-off-closure}
    A class $\C$ of graphs is closed under taking split-off minors if and only if there is a class $\M$ of multigraphs whose restriction to simple graphs is equal to $\C$ and $\M$ is closed under taking (multigraph) minors and immersions.
\end{proposition}
\begin{proof}
    It follows directly from the definition of split-off minors that a class of multigraphs is closed under taking split-off minors if and only if it is closed under taking minors and immersions. \par 
    If $\M$ is a class of multigraphs that is closed under taking split-off minors then its restriction to simple graphs is also closed under taking split-off minors (among the simple graphs). For the other direction, we can define $\M$ as the class of multigraphs that are split-off minors of an element of $\C$.
\end{proof}

By using Corollary \ref{cor:som-means-hdclosed} we manage to recover many of the known examples of classes that are closed under weak oddomorphisms, and some new ones.
\begin{example}
    If $G$ is obtained by splitting off a pair of edges in a tree $F$, then $G$ is a disjoint union of minors of $F$ (even topological minors of $F$). Thus it follows that every class of forests that is closed under disjoint union and taking minors is closed under taking split-off minors. It then follows from Corollary \ref{cor:som-means-hdclosed} that every class of forests that is minor closed and closed under taking disjoint unions is h.d.\ closed.
\end{example}
Note that a slightly stronger statement was proved by Neuen and Seppelt \cite{neuen2026distinguishinggraphscountinghomomorphisms}, namely that every class of forests closed under taking \emph{topological minors} and disjoint unions is homomorphism distinguishing closed.

\begin{example}
    Each component of a split-off minor of a cycle or a path of length $n$ is a cycle or a path of length at most $n$. It thus follows that the class of disjoint unions of cycles and paths of length less than $m$ for $m \in \N \cup \{\infty\}$ is h.d.\ closed.
\end{example}
It should be noted that the result from the above example can be deduced from results by Roberson \cite{roberson2022oddomorphisms}.
\par 
\begin{example}
    A multigraph is a \emph{cactus graph} if each edge belongs to at most one cycle.
    It is an easy exercise to show that cactus multigraphs are closed under both taking minors and immersions.
    Thus it follows that the class of simple cactus graphs is closed under taking split-off minors. Since it is also closed under taking disjoint unions, Corollary \ref{cor:som-means-hdclosed} implies that it is h.d.\ closed.
\end{example}
Note that the h.d.\ closedness of the class of simple cactus graphs was already established by Neuen and Seppelt \cite{neuen2026distinguishinggraphscountinghomomorphisms}.
\par 
Observe that the class of disjoint unions of cycles and paths is the class of graphs that exclude the star of degree 3 as a split-off minor.
Likewise, the class of forests can be defined as the class of graphs that exclude the double edge as a split-off minor, and the class of cactus graphs is the class of graphs that exclude the triple edge as a split-off minor.
It of course also follows from Corollary \ref{cor:thm6wqo-variant} that the class of graphs that exclude the $n$-fold edge as a split-off minor is also h.d.\ closed. 
This gives a hierarchy of classes defined by a connectivity property that includes forests and cactus graphs, all of which are now shown to be h.d.\ closed.
These classes can of course also be described by excluding a finite set of minors, or a finite set of immersions. But there is no reason to believe that these finite sets admit a simple description.
\par 
Lastly, we use Corollary \ref{cor:som-means-hdclosed} to identify a hierarchy of classes of graphs that have not previously been shown to be h.d.\ closed. 
\begin{definition}
    The \emph{cyclomatic number} of a multigraph $G$ is defined as $\abs{E(G)} - \abs{V(G)} + \abs{c(G)}$ where $c(G)$ is the set of connected components of $G$.    
\end{definition}
Equivalently, the cyclomatic number of $G$ can be defined as the minimum number of edges that need to be deleted from $G$ to make it a forest.

\begin{corollary}
\label{cor:cyclomatic}
    The class of graphs whose connected components have cyclomatic number at most $n$ is h.d.\ closed.
\end{corollary}
\begin{proof}
    We show that the class is closed under taking split-off minors. The class is trivially closed under taking disjoint unions. It thus follows from Corollary~\ref{cor:som-means-hdclosed} that it is h.d.\ closed.
    To show that the class is closed under taking split-off minors, it suffices to show that when any of the allowed operations are applied to a connected graph $G$, the cyclomatic number of each of the resulting components is at most the cyclomatic number of $G$. We first show that the graph $F$ obtained from $G$ by applying one of the allowed operations has cyclomatic number at most that of $G$. We examine the four cases separately.
    \begin{itemize}
        \item If $F$ is obtained from $G$ by contracting an edge, then the numbers of edges and vertices are both reduced by one, while the number of connected components stays constant. Hence the cyclomatic number of $F$ is equal to the cyclomatic number of $G$.

        \item If $F$ is obtained from $G$ by splitting off a pair of edges, then the number of edges is reduced by one, the number of vertices is unchanged, and the number of connected components increases by at most one. Thus the cyclomatic number of $F$ is at most the cyclomatic number of $G$.

        \item If $F$ is obtained from $G$ by deleting an edge, then the number of edges decreases by one, the number of vertices is unchanged, and the number of connected components increases by at most one. Thus the cyclomatic number of $F$ is at most the cyclomatic number of $G$.

        \item If $F$ is obtained from $G$ by deleting an isolated vertex, then the number of edges is unchanged, while the numbers of vertices and connected components both decrease by one. Hence the cyclomatic number of $F$ is equal to the cyclomatic number of $G$.
    \end{itemize}
    Now note that the cyclomatic number of $F$ is equal to the sum of the cyclomatic numbers of its components. Thus the cyclomatic number of each component of $F$ is at most the cyclomatic number of $G$, which finishes the proof.
\end{proof}

We should note that the preceding result resembles a result by Neuen and Seppelt \cite{neuen2026distinguishinggraphscountinghomomorphisms}, which states that the class of graphs whose connected components have a \emph{feedback vertex set number} at most $n$ is h.d.\ closed. 
The feedback vertex set number of a graph is the minimum number of vertices that need to be deleted from the graph to make it a forest.

\section{Conclusion}

In this paper, we explored the connection between oddomorphisms and minor-like structural graph relations.
We proved that the existence of an oddomorphism from a graph $F$ to a graph $G$ does not imply that $G$ is a minor of $F$, but it does imply that $G$ is a split-off minor of $F$.
\par 
A natural variant of Question~\ref{q} asks whether the existence of an oddomorphism from $F$ to $G$ implies that $F$ contains $G$ as an immersion. A positive answer would strengthen our result about split-off minors and imply the immersion variant of the Strong Roberson Conjecture. This question has since been resolved in the negative in a recent preprint~\cite{kristjansson2026counterexamples}, which also refutes the Strong Roberson Conjecture and its immersion variant.
\par  
Another variant which remains open is whether an oddomorphism from $F$ to $G$ implies that $G$ is a lift-minor of $F$. 
For connected $G$, the argument in the proof of Theorem~\ref{thm:oddo-implies-lm} can be adapted to obtain an odd lift $F'$ of $G$ as a lift-minor of $F$.
It therefore suffices to establish a lift-minor analogue of Lemma~\ref{lemma:perfect-matching-lm} for odd lifts.
\par
It also remains open whether every $l$-lift of $K_n$ contains $K_n$ as an immersion. Following Drier and Linial~\cite{Drier2004Minors}, it would also be interesting to determine the size of the largest clique that immerses into a random $l$-lift of $K_n$. This size is at most $n$ and, by~\cite{GAUTHIER201998}, is $\Omega(n)$.
A result of DeVos, Kawarabayashi, Mohar, and Okamura~\cite{devos2010immersing} implies that for $n\leq 7$, every $l$-lift of $K_n$ contains a $K_n$ immersion.
\par 
It also remains open whether the examples in Section~\ref{sec:oddo-doesnt-mean-minor} lead to further classes closed under minors and disjoint unions that are not h.d.\ closed. The corresponding classes are candidates because they are not closed under weak oddomorphisms, a necessary condition for such a class to fail to be h.d.\ closed. For example, it is not known whether the class of graphs of Hadwiger number at most $n$ is h.d.\ closed for large $n$.

\bibliography{bibtex}

\end{document}